\documentclass[article,showpacs, twocolumn]{revtex4-1}

\usepackage{amsmath}
\usepackage{amsfonts}
\usepackage{dsfont}
\usepackage{graphicx}
\usepackage{caption}
\usepackage{bbm}
\usepackage{bbold}
\usepackage{epstopdf}
\usepackage[font=scriptsize]{caption}
\usepackage{color}
\usepackage{amsthm}
\usepackage{chngcntr}

\newtheorem{theorem}{Theorem}
\newtheorem{lemma}{Lemma}
\newtheorem{corollary}{Corollary}

\counterwithin{lemma}{section}
\counterwithin{corollary}{section}

\newcommand{\ket}[1]{\ensuremath{\left| #1 \right\rangle}}
\newcommand{\br}[1]{\ensuremath{\left\langle #1 \right.}}
\newcommand{\bra}[1]{\ensuremath{\left. \br{#1} \right|}}

\newcommand{\kb}[2]{\ket{{#1}}\bra{{#2}}}

\newcommand{\proj}[1]{\kb{{#1}}{{#1}}}

\newcommand{\tr}{{\rm Tr}}

\definecolor{red}{rgb}{0.9,0,0}
\definecolor{green}{rgb}{0,0.8,0}
\definecolor{blue}{rgb}{0,0,0.8}
\definecolor{cautionred}{rgb}{1.0,0,0}

\begin{document}
\title{The thermodynamic cost of quantum operations}
\date{\today}
\author{D.~J.~Bedingham}
\email{daniel.bedingham@philosophy.ox.ac.uk}
\author{O.~J.~E.~Maroney}
\email{owen.maroney@philosophy.ox.ac.uk}
\affiliation{Faculty of Philosophy, University of Oxford, OX2 6GG, United Kingdom.}
\begin{abstract}
The amount of heat generated by computers is rapidly becoming one of the main problems for developing new generations of information technology. The thermodynamics of computation sets the ultimate physical bounds on heat generation. A lower bound is set by the Landauer Limit, at which computation becomes thermodynamically reversible. For classical computation there is no physical principle which prevents this limit being reached, and approaches to it are already being experimentally tested. In this paper we show that for quantum computation there is an unavoidable excess heat generation that renders it inherently thermodynamically irreversible.  The Landauer Limit cannot, in general, be reached by quantum computers.  We show the existence of a lower bound to the heat generated by quantum computing that exceeds that given by the Landauer Limit, give the special conditions where this excess cost may be avoided, and show how classical computing falls within these special conditions.
\end{abstract}

\maketitle
%%%%%%%%%%%%%%%%%%%%%%%%%%%%%%%%%%%%%%%%%%%%
\section{Introduction}
%%%%%%%%%%%%%%%%%%%%%%%%%%%%%%%%%%%%%%%%%%%%

Information processing does not come for free.  Physical systems are needed to store, transmit and process information, and these come with physical resource costs, in time, space and energy.  The heat generated by information processing is becoming one of the most significant of these costs. Modern integrated circuits have power densities of the order $100 {\rm W/cm}^{2}$ with operating temperatures near to their upper limit \cite{VASS}. Thermal management is already the main constraint on the performance of modern electronics and advances in technology will increasingly need to test the fundamental limits of energy consumption.

For classical information processing, Landauer's Principle \cite{LAN,BEN1,LR2003,Maroney2009} relates the change in information from a computation to a minimum thermodynamic cost in the form of heat generated in the environment.  The information content of a source of signals is measured by the Shannon measure $H=-\sum_n p_n \log_2 p_n$, where $p_n$ is the probability of signal $n$ occurring.  Landauer's Principle states that the mean heat generated by an information processing device is given by
\begin{align}
\Delta E_{\text{classical}} \geq  - k T \ln 2 \Delta H,
\label{clp}
\end{align}
where $\Delta H$ is the change in the Shannon information.

The limit cost, which corresponds to the condition of thermodynamic reversibility \cite{OM1,SAG1}, is defined purely in terms of the information processing operation itself: every physical system which performs the task must pay at least this cost \cite{Pie00,OM1,Turgut2006}.  Most importantly, it is a tight bound: while there are practical barriers to reaching the limit (finite size and time effects, single shot costs, etc.~\cite{Dahlstein2013,RW2014}) there is no physical principle that prevents their effects becoming arbitrarily small, with experimental tests increasingly pushing at this boundary \cite{BER,JUN,MARTINI}. 
This tight relationship between information theory and thermodynamics makes Landauer's Principle the basis of our understanding of the thermodynamic cost of information processing.

However, both the theoretical studies and experimental tests of the Landauer Limit have almost exclusively considered classical information processing. Quantum computing radically alters the nature of information, leading to many attempts to extend Landauer's Principle to quantum theory \cite{VV,AND,GOO,AT,SAG2,KIM,SAG3,GOO2}.

Quantum information processing involves a quantum system ${\rm S}$ with states $\rho^n_{\rm S}$ as signal states, undergoing a quantum operation $Q$, defined by its effect upon the signal states $\rho^n_{\rm S} \rightarrow \rho^{n\prime}_{\rm S} = \mathcal{Q}(\rho^n_{\rm S})$. Given that the input $\rho^n_{\rm S}$ appears with probability $p_n$, the average input is  $\rho_{\rm S}=\sum_n p_n \rho^n_{\rm S}$, and the average output is $\rho_{\rm S}'=\sum_n p_n \rho^{n\prime}_{\rm S}$.

The natural quantum generalisation of Shannon information to quantum theory is the Schumacher information measure $S(\rho_{\rm S})=-\tr [ \rho_{\rm S}\log_2 \rho_{\rm S} ]$ \cite{Sch95}, which agrees with the Shannon information when the signal states are orthogonal.  This suggests that Equation~(\ref{clp}) should be replaced by
\begin{align}
\Delta E_{Q} \geq  k T \ln 2 \left[S(\rho_{\rm S})- S(\rho_{\rm S}')\right].
\label{bound}
\end{align}
While it is well known that  Equation~(\ref{bound}) must hold, and that reaching the equality is required for thermodynamic reversibility \cite{PAV,Maroney2007a}, we show that it fails as a tight bound for quantum information processing.  There is an unavoidable excess heat generation in quantum computation that is not present in the classical case.

%%%%%%%%%%%%%%%%%%%%%
\subsection*{Summary of the theorem}
%%%%%%%%%%%%%%%%%%%%%

For any given quantum operation $Q$ there is non negative quantity $\epsilon_Q $, defined independently of how the quantum operation is physically implemented, such that the heat generated by the operation is bounded by:
\begin{align}
\Delta E_Q \geq  kT \ln 2\left[ S(\rho_{\rm S}) - S(\rho'_{\rm S})\right] + \epsilon_Q.
\nonumber
\end{align}
In general for quantum operations $\epsilon_Q >0$.  Barring exceptional symmetric cases, $\epsilon_Q =0$ if, and only if, all the output signal states $\mathcal{Q}(\rho^n_{\rm S})$ share a common diagonalised basis, and there exists a common left stochastic map from diagonal elements of the input states (in the diagonal basis of $\rho_{\rm S}$) to eigenvalues of the output states.  If these conditions do not hold, there is an excess, thermodynamically irreversible cost to any physical process that performs the quantum operation.  

\medskip

This paper will define quantum operations, characterise the thermodynamics of implementing quantum operations, state the theorem, and then give examples of operations which meet, and operations which do not meet, the conditions for $\epsilon_Q=0$.

%%%%%%%%%%%%%%%%%%%%%%%%%%%%%%%%%%%%%%%%%%%%
\section{Quantum Operations}
%%%%%%%%%%%%%%%%%%%%%%%%%%%%%%%%%%%%%%%%%%%%
While classical information processing is built up from logical operations such as AND, OR and NOT gates, the fundamental element in quantum information processing is the quantum operation \cite{NC}.  This acts upon one of a number of possible input states, each represented by a different quantum density matrix, $\{\rho^n_{\rm S}\}$, and maps each one to a specific output state $\rho^n_{\rm S} \rightarrow {\rho^{n\prime}_{\rm S}}$.

Not all maps of the form $\rho^n_{\rm S} \rightarrow {\rho^{n\prime}_{\rm S}}$ are physically possible.  A quantum operation, $Q$, must be a completely positive, trace preserving linear map: $\mathcal{Q}(\rho^n_{\rm S})=\sum_k Q_k \rho^n_{\rm S} Q^\dag_k$, with $\sum_k Q_k Q_k^\dag=\mathbb{1}_{\rm S}$.    When it is not a pure unitary rotation, a quantum operation requires the use of an auxiliary system ${\rm A}$, initially in a standard state, $\rho_{\rm A}$.  A joint unitary $V$ acting on $\rho^n_{\rm S}\otimes \rho_{\rm A}\rightarrow \rho_{{\rm SA}}^{n\prime}=V\rho^n_{\rm S}\otimes \rho_{\rm A} V^{\dagger}$ results in an entangled state of system and auxiliary (see Figure \ref{F0}). The reduced state of the system, ${\rho^{n\prime}_{\rm S}} = \tr_{\rm A} [V\rho^n_{\rm S}\otimes \rho_{\rm A} V^{\dagger}]$, is the output state (see \cite{NC} Chapter 8, for a textbook presentation).

\begin{figure}
        \begin{center}
        	\includegraphics[width=0.3\textwidth]{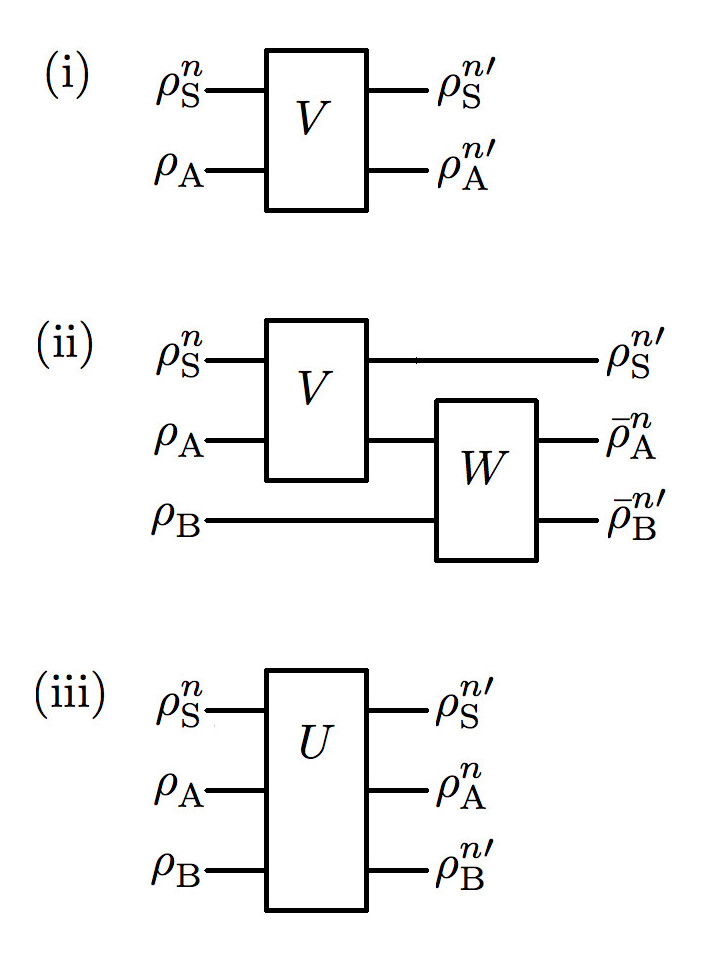}
        \end{center}
\caption{(i) The basic operation we consider involves a unitary $V$ to convert an input system state $\rho^{n}_{\rm S}$  into an output state $\rho^{n\prime}_{\rm S}$. An auxiliary input $\rho_{\rm A}$ typically becomes entangled with the system and is discarded at the end of the operation. In order to understand the energy cost of the operation we should reset the auxiliary to its standard state $\rho_{\rm A}$ in preparation for further uses. This is done via interaction with a heat bath which starts in the canonical state $\rho_{\rm B}$. (ii) An energetically suboptimal way to do this ignores any correlations between ${\rm S}$ and ${\rm A}$ and directly resets the auxiliary using the unitary $W$ which acts on the joint auxiliary-heat bath state. The output auxiliary state $\bar{\rho}_{\rm A}^n = \tr_{\rm B}[W\rho^{n\prime}_{\rm A}\otimes\rho_{\rm B}W^{\dagger}]$ should be such that ${\rho}_{\rm A} = \sum_n p_n \bar{\rho}_{\rm A}^n$ where $p_n$ is the probability of input $n$; the auxiliary is reset on average. (iii) The optimal method typically cannot be decomposed in this way and requires a unitary $U$ defined to act on system, auxiliary, and heat bath. $U$ should result in the same output state $\rho^{n\prime}_{\rm S}$ and reset the auxiliary: the output auxiliary state ${\rho}_{\rm A}^n$ should be such that $\rho_{\rm A} = \sum_n p_n {\rho}_{\rm A}^n$. For given average input $\rho_{\rm S}$ and output $\rho'_{\rm S}$, the average energy change which occurs in the heat bath will be lower bounded.
}
\label{F0}
\end{figure}

%%%%%%%%%%%%%%%%%%%%%%%%%%%%%%%%%%%%%%%%%%%%
\section{Thermodynamics of quantum operations}
%%%%%%%%%%%%%%%%%%%%%%%%%%%%%%%%%%%%%%%%%%%%
It is essential, when quantifying the thermodynamic costs of operations, to keep track of the effects upon auxiliary systems.  If an auxiliary is disregarded, then the quantum operation would appear to be performed without generating heat.  However, this leaves the auxiliary in a state ${\rho_{\rm A}^{n\prime}}=\tr_{\rm S} [V\rho^n_{\rm S} \otimes \rho_{\rm A} V^\dag]$.  Simply discarding the auxiliary would, on average, change the entropy of the environment by $\ln 2 \left[S(\rho_{\rm A}') - S(\rho_{\rm A})\right]$ where $\rho_{\rm A}' = \tr_{\rm S} [V\rho_{\rm S}\otimes \rho_{\rm A} V^{\dagger}]=\sum_n p_n {\rho_{\rm A}^{n\prime}}$.

The most straightforward way to deal with this is to reset the auxiliary to its original state $\rho_{\rm A}$. To do this a thermal environment is introduced in the form of a heat bath in a canonical state $\rho_{\rm B}$ at temperature $T$. The reset can be performed by some unitary $W$ acting on both the auxiliary and the heat bath such that $\rho_{\rm A}=\tr_{\rm B}[W\rho'_{\rm A}\otimes \rho_{\rm B} W^\dag]$ (see Figure \ref{F0}). This operation would transfer at least $kT \ln 2 \left[S(\rho'_{\rm A}) - S(\rho_{\rm A})\right]$ of heat to the heat bath \cite{PAV}.  Either way, the cost to the environment will generally be more than required by Equation~(\ref{bound}).  A quantum operation typically leaves correlations between system and auxiliary, with a mutual information \cite{VV} of $S(\rho^\prime_{\rm S}\colon\rho^\prime_{\rm A})=\left[S(\rho'_{\rm S}) - S(\rho_{\rm S})\right] + \left[S(\rho_{\rm A}') - S(\rho_{\rm A})\right] \geq 0$.  Simply resetting (or discarding) the auxiliary will pay $k T \ln 2 S(\rho_{\rm S}' \colon\rho_{\rm A}')$ as an excess cost unless there is no correlation: $V\rho_{\rm S}\otimes \rho_{\rm A} V^{\dagger}=\rho_{\rm S}' \otimes \rho_{\rm A}'$.

To find the minimum thermodynamic cost of the quantum operation, the auxiliary must be reset more efficiently, exploiting correlations with the system. The unitary $V$ must be embedded in a larger unitary $U$ which includes interactions between system, auxiliary, and heat bath. $U$ must preserve the output of the computation
\begin{align}
\tr_{\rm AB} [U\rho^n_{\rm S}\otimes\rho_{\rm A}\otimes\rho_{\rm B}U^{\dagger}] &= {\rho^{n\prime}_{\rm S}};
\label{Ereset0}
\end{align}
while resetting the auxiliary
\begin{align}
\tr_{\rm SB} [U\rho_{\rm S}\otimes\rho_{\rm A}\otimes\rho_{\rm B}U^{\dagger}] &= \rho_{\rm A}
\label{Ereset}
\end{align}
(see Figure \ref{F0}). 
Standard calculations show that this implies Equation (\ref{bound}) (see Appendix \ref{Sqtd}). If the quantum operation is only defined for a single signal state, then physical implementations are possible which can get arbitrarily close to the equality in (\ref{bound}) (Appendix \ref{ARev}).

In general, a quantum operation must produce the result ${\rho^{n\prime}_{\rm S}}=\mathcal{Q}(\rho^n_{\rm S})$, for multiple signal states $\{\rho_{\rm S}^n\}$. The same $U$ must satisfy (\ref{Ereset0}) for all $n$.  Our principal result is to show this additional constraint forces a higher thermodynamic cost than Equation~(\ref{bound}). An operation which generates a quantity of heat, $\epsilon_Q$, in excess of this bound will be thermodynamically irreversible, as any second operation which restores the original average state $\rho_{\rm S}=\mathcal{Q'}(\rho_{\rm S}')$ must necessarily leave a net heat gain in the heat bath of at least $\epsilon_Q$ over the complete cycle. Quantum computations cannot, in general, be performed in a thermodynamically reversible manner.

%%%%%%%%%%%%%%%%%%%%%%%%%%%%%%%%%%%%%%%%%%%%
\section{The Theorem}
%%%%%%%%%%%%%%%%%%%%%%%%%%%%%%%%%%%%%%%%%%%%

Let $\rho^n_{\rm S}\rightarrow \rho^{n\prime}_{\rm S} = \tr_{\rm AB} [U \rho^n_{\rm S}\otimes \rho_{\rm A}\otimes\rho_{\rm B} U^{\dagger}]$ be a quantum operation $Q$ where $\rho_{\rm B}$ is the state of a canonical heat bath at temperature $T$, and $\rho_{\rm A}$ is a standard state of an auxiliary which should be restored by the operation [see Equation~(\ref{Ereset})].

The average system input state $\rho_{\rm S}$ can be expressed in terms of its diagonal basis vectors $\{|\phi_i\rangle\}$ as
\begin{align}
\rho_{\rm S} = \sum_n p_n \rho^n_{\rm S} =\sum_i \lambda_i |\phi_i\rangle\langle \phi_i |.
\nonumber
\end{align}
Similarly the average system output state can be expressed in terms of its diagonal basis vectors $\{|\phi'_k\rangle\}$ as
\begin{align}
\rho'_{\rm S} = \sum_n p_n \rho^{n\prime}_{\rm S} =\sum_k \lambda'_k |\phi'_k\rangle\langle \phi'_k|.
\nonumber
\end{align}
In terms of these basis vectors the individual inputs and outputs are represented by
\begin{align}
\rho^n_{\rm S} =  \sum_{ij} \mu_{ij}^n |\phi_i\rangle\langle \phi_j | \quad \text{and} \quad
\rho^{n\prime}_{\rm S} =  \sum_{kl} \mu_{kl}^{n\prime}|\phi'_k\rangle\langle \phi'_l |.
\nonumber
\end{align}
We now state the theorem:
\begin{theorem}
\label{th1}
If there exists a stochastic map, $P_Q(k|i)$ with
\begin{align}
\sum_k P_Q(k|i) = 1 \quad \text{and} \quad P_Q(k|i) \geq 0 \; \forall i,k, \nonumber
\end{align}
such that
\begin{align}
\mu^{n\prime}_{kl} = \delta_{kl} \sum_i P_Q(k|i) \mu^n_{ii},
\label{revcon}
\end{align}
for all $n$, then the minimum thermodynamic cost of the operation
\begin{align}
\Delta E_Q \geq kT \ln 2 \left[ S(\rho_{\rm S}) - S(\rho'_{\rm S})\right]+\epsilon_Q,
\nonumber
\end{align}
can approach $\epsilon_Q=0$. Otherwise, provided there are no symmetries of the form $\lambda_i/\lambda_j=\lambda'_k/\lambda'_l$ where $i\neq j$ or $k\neq l$, then necessarily $\epsilon_Q>0$.
\end{theorem}

Note that the set of properties $\{\{\lambda_i\},\{\lambda'_k\},\{\mu_{ij}^n\},\{{\mu_{kl}^{n\prime}}\}\}$, used to state the theorem are defined solely in terms of the quantum operation, independently of the particular physical process used to implement it.

\subsection*{Outline of proof}
In Appendix \ref{NAFS} we show that the implementation of the operation can be re-written in the form
\begin{align}
\mu^{n\prime}_{kl}=\sum_{ij} q(kl|ij) \mu^n_{ij},
\label{mapp}
\end{align}
where the complex coefficients $q(kl|ij)$ carry the effects of interaction with the environment. They have the properties $q(kk|ii)=q^*(kk|ii)\geq 0$, and $\sum_k q(kk|ij)=\delta_{ij}$. In Appendix \ref{Sqtd} we show that for a given physical implementation of the quantum operation where $\rho'= U \rho_{\rm S}\otimes \rho_{\rm A}\otimes \rho_{\rm B} U^{\dagger}$, and $\rho_\star=\rho'_{\rm S} \otimes \rho_{\rm A}\otimes\rho_{\rm B} $,
\begin{align}
\epsilon_Q\geq \frac{1}{2}kT|| \rho'-\rho_\star||_1^2.
\label{newbound}
\end{align}
Finally in Appendix \ref{SUR} we show that
\begin{align}\label{eqbound1}
|| \rho'-\rho_\star||_1 \geq \frac{|\lambda_j\lambda'_k - \lambda_i\lambda'_l|}{\lambda'_k+\lambda'_l} \left| q(kl|ij) \right|.
\end{align}
It follows that if there is a value of $\left| q(kl|ij) \right|>0$ for which $\lambda_i/\lambda_j \neq \lambda'_k/\lambda'_l$, then the implementation has an excess thermodynamic cost.

The coefficients $q(kk|ii)$ imply no bound as $\lambda_i/\lambda_j = \lambda'_k/\lambda'_l$ automatically holds for them. If Equation~(\ref{revcon}) holds, then $q(kk|ii) = P_Q(k|i)$ gives an implementation with no excess cost (see Appendix \ref{proto} for an explicit construction).  Otherwise, there must be some $\left|q(kl|ij)\right|>0$ for $i \neq j$ or $k \neq l$.  If $\lambda_i/\lambda_j \neq \lambda'_k/\lambda'_l$ for any of these $ijkl$ values, then there is an excess cost for that implementation.  We can therefore use numerical optimisation techniques to find the coefficients satisfying (\ref{mapp}) which minimise the largest value of (\ref{eqbound1}), and Equation~(\ref{newbound}) shows this gives a lower bound for $\epsilon_Q$.

When the output states do not have a shared diagonalised basis, so that there is some $\mu^{n\prime}_{kl}>0$ for $k \neq l$, there is also an analytical bound:
\begin{align}
\epsilon_Q  \geq \frac{1}{2}kT
\max_{n,k,l\neq k}\left\{
\frac
{|\mu^{n\prime}_{kl}|/(\lambda'_k+\lambda'_l)}
{\sum_{ij} |\mu^n_{ij}|/|\lambda_j\lambda'_k - \lambda_i\lambda'_l|}
\right\}^2. \label{eqoffdiagM}
\end{align}
For further details see Appendix \ref{secPROOF}.

This energy bound and others derived here may not be tight. It is an open problem to demonstrate a protocol which reaches $\epsilon_Q$ - our results only place a non-zero lower bound on the excess energy cost.

\begin{figure*}
        \begin{center}
        	\includegraphics[width=0.8\textwidth]{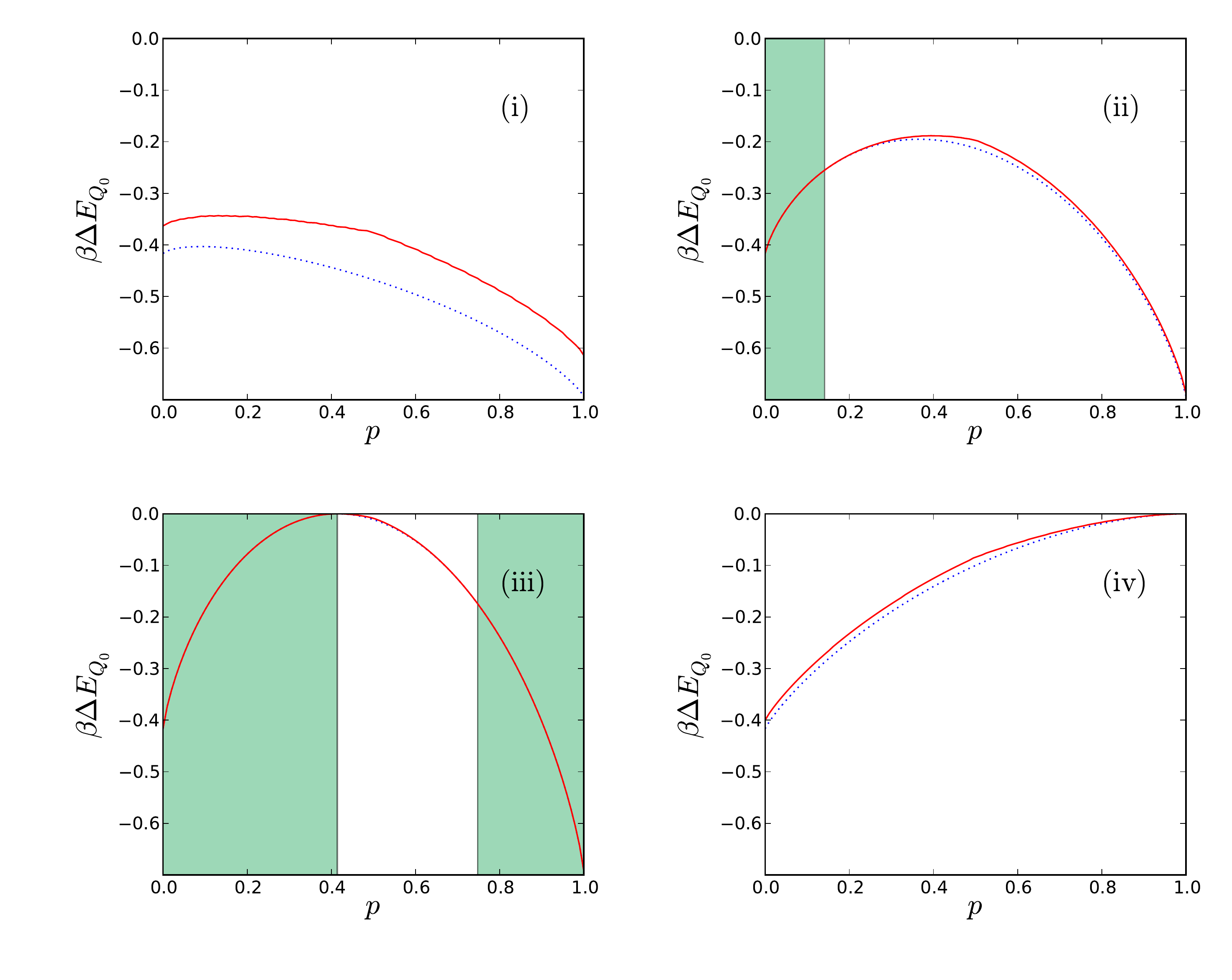}
        \end{center}
\caption{The figure shows the minimum energy change (scaled by $\beta = 1/kT$) in the heat bath environment as a result of a qubit dephasing operation $Q_0$. Notice that for the dephasing operation the energy cost is negative - the operation draws energy from the heat bath. This means that the operation can be used to perform work. In each case the inputs of the operation are $|v^1_{\rm S}\rangle$ with probability $p$ and $|v^2_{\rm S}\rangle $ with probability $1-p$. Each panel shows the energy change as a function of $p$ for a different pair of input vectors: (i) input vectors $|v^1_{\rm S}\rangle = |+\rangle= \frac{1}{\sqrt{2}}|0\rangle + \frac{1}{\sqrt{2}}|1\rangle $  and $|v^2_{\rm S}\rangle =\cos\frac{\pi}{8}|0\rangle + \sin\frac{\pi}{8}|1\rangle$; (ii) input vectors $|v^1_{\rm S}\rangle = \frac{1}{\sqrt{2}}|0\rangle + i \frac{1}{\sqrt{2}}|1\rangle  $  and $|v^2_{\rm S}\rangle = \cos\frac{\pi}{8}|0\rangle + \sin\frac{\pi}{8}|1\rangle$; (iii) input vectors $|v^1_{\rm S}\rangle = |-\rangle=\frac{1}{\sqrt{2}}|0\rangle - \frac{1}{\sqrt{2}}|1\rangle  $  and $|v^2_{\rm S}\rangle = \cos\frac{\pi}{8}|0\rangle + \sin\frac{\pi}{8}|1\rangle$; (iv) input vectors $|v^1_{\rm S}\rangle = |0\rangle $  and $|v^2_{\rm S}\rangle = \cos\frac{\pi}{8}|0\rangle + \sin\frac{\pi}{8}|1\rangle$. The blue dotted lines indicate the thermodynamically reversible bound on the energy cost given by Equation~(\ref{bound}); the red lines indicate the lower bounds determined in this article. The green bands indicate where the bound on the energy cost satisfies the condition of thermodynamic reversibility. Within these regions there exists a protocol to achieve the reversible energy cost.
}
\label{F1}
\end{figure*}

\begin{figure*}
        \begin{center}
        	\includegraphics[width=1.0\textwidth]{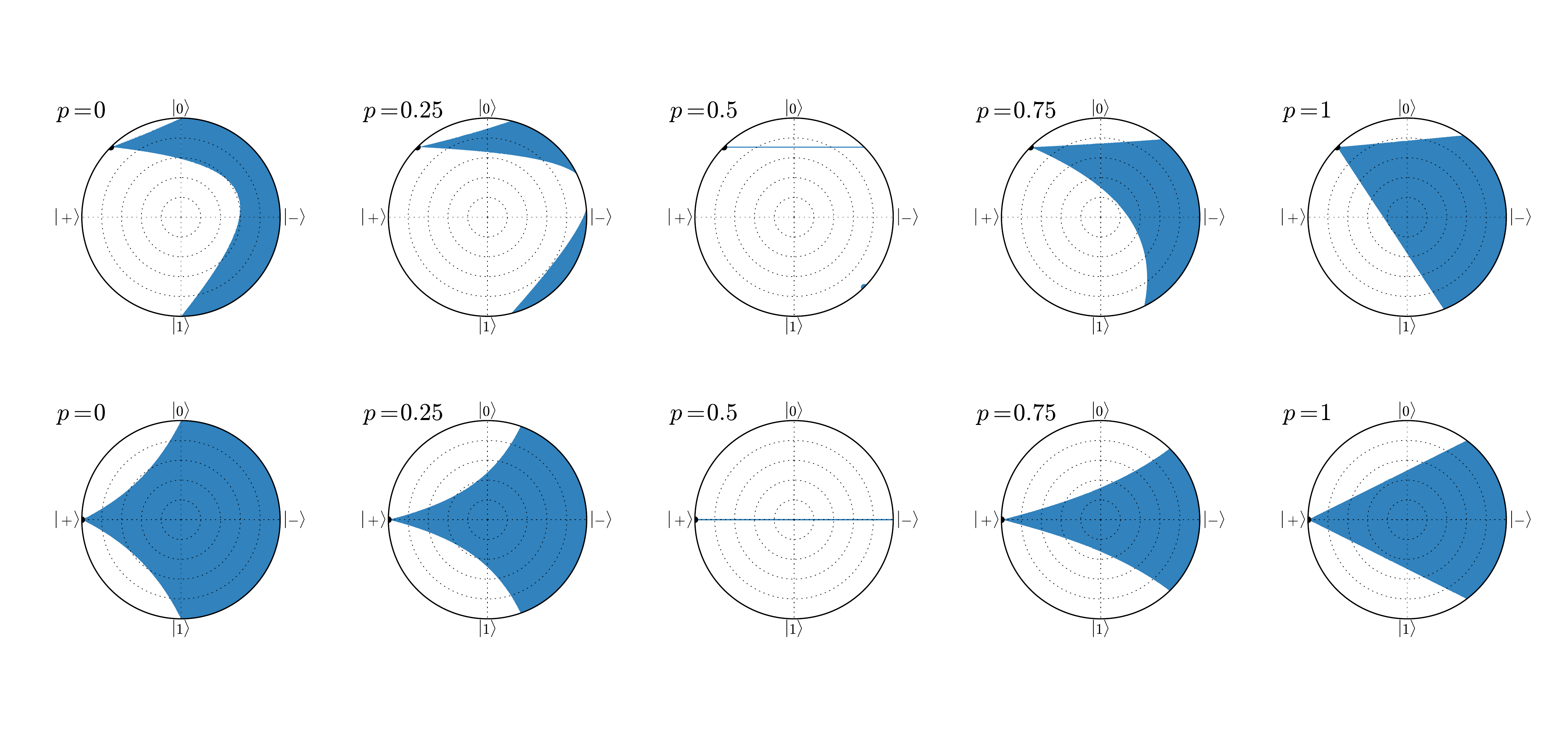}
        \end{center}
\caption{The figure shows the Bloch sphere from one side. For a given pure state input $\rho_{\rm S}^1 = |v_{\rm S}^1\rangle\langle v_{\rm S}^1 |$, occurring with probability $p$, each figure shows the regions of a second pure state input $\rho_{\rm S}^2 = |v_{\rm S}^2\rangle\langle v_{\rm S}^2 |$, occurring with probability $p-1$, such that the condition of thermodynamic reversibility can be satisfied for a qubit dephasing operation $Q_0$. In the top row, the fixed input vector is $|v^1_{\rm S}\rangle = \cos\frac{\pi}{8}|0\rangle + \sin\frac{\pi}{8}|1\rangle$ (marked by a black dot); in the bottom row the fixed input vector is $|v^1_{\rm S}\rangle = |+\rangle=\frac{1}{\sqrt{2}}|0\rangle + \frac{1}{\sqrt{2}}|1\rangle$. From left to right the value of $p$ varies from $0$ to $1$. For given $p$ the shaded region contains the vectors $|v^2_{\rm S}\rangle$ for which $\epsilon_{Q}=0$ and the operation can be thermodynamically reversible. When $p=0.5$ this region contains only the points on the same circle of latitude as $|v^1_{\rm S}\rangle$ and the point opposite $|v^1_{\rm S}\rangle$. When $p=0$ the shaded region does not correspond to the entire Bloch sphere since we still demand that the implementation should give that correct output for both possible inputs (similarly for $p=1$). The images are symmetric when viewed from the opposite hemisphere. Also, if we include a phase shift to the vector $|v^1_{\rm S}\rangle$ the regions rotate around the Bloch sphere with $|v^1_{\rm S}\rangle$, maintaining their form.
}
\label{F2}
\end{figure*}

%%%%%%%%%%%%%%%%%%%%%%%%%%%%%%%%%%%%%%%%%%%%
\section{Examples of thermodynamically {reversible} operations}
%%%%%%%%%%%%%%%%%%%%%%%%%%%%%%%%%%%%%%%%%%%%
There are well known cases of quantum operations where $\epsilon_Q=0$.  We now show how they fit into our proof.

\textit{Pure unitary.}
In a pure unitary quantum operation, $\rho^{n\prime}_{\rm S}=U\rho^n_{\rm S} U^\dag$.  This does not involve an auxiliary or heat bath, and does not change the eigenvalues of any input signal states.  Solutions with $q(kl|ij)=\delta_{ik}\delta_{jl}$ satisfy the operation, and the exceptional symmetry condition applies: the only terms which contribute to Equation~(\ref{mapp}) are ones for which $\lambda_i/\lambda_j=\lambda'_k/\lambda'_l$.

\textit{Single input.}
When there is only one possible input $\rho^n_{\rm S} = \rho_{\rm S}$. As $\mu^n_{ij}=\delta_{ij} \lambda_i$ and $\mu^{n\prime}_{kl}=\delta_{kl}\lambda'_k$, choosing $P_Q(k|i) = \lambda'_k$ shows that Equation~(\ref{revcon}) holds.

\textit{Single output.} Resetting to a standard state is known to satisfy $\epsilon_Q=0$ \cite{SAG2,AT}.  In this case the output states for all inputs are of the form $\rho'_{\rm S} = \sum_k \lambda'_k \proj{\phi_k}$, so $\mu^{n\prime}_{kl}=\delta_{kl}\lambda'_k$.  Again, choosing $P_Q(k|i)=\lambda'_k$, shows that Equation~(\ref{revcon}) holds.

\textit{Classical information processing.}
For classical information processing, where the signal states correspond to orthogonal quantum states, $\epsilon_Q=0$ \cite{Mar05b,OM1,SAG1}.  $\rho^n_{\rm S} =  |\phi_n\rangle\langle \phi_n |$, so $\mu^n_{ij}=\delta_{ij} \delta_{in}$.  In general the outputs take the form $\rho^{n\prime}_{\rm S} =  \sum_{k} P_Q(k|n) |\phi'_k\rangle\langle \phi'_k |$, representing a distribution of possible outputs $|\phi'_k\rangle$, so $\mu^{n\prime}_{kl}=\delta_{kl} P_Q(k|n)$, and Equation~(\ref{revcon}) holds.

%%%%%%%%%%%%%%%%%%%%%%%%%%%%%%%%%%%%%%%%%%%%
\section{Case study of a thermodynamically irreversible operation}
%%%%%%%%%%%%%%%%%%%%%%%%%%%%%%%%%%%%%%%%%%%%
In general it will not be the case that $\epsilon_Q=0$.
Even the simplest quantum operations will fail to satisfy Equation~(\ref{revcon}). As an example, we consider single qubit dephasing operations, which exhibit all the features of the Theorem.  Inputs of the form $\rho^n_{\rm S} = |v^n_{\rm S}\rangle\langle v^n_{\rm S}|$ with $|v^n_{\rm S}\rangle = c_0^n|0_{\rm S}\rangle + c_1^n|1_{\rm S}\rangle$ give outputs $\rho^{n\prime}_{\rm S}(r) =r\rho^n_{\rm S}+ (1-r) \left[ |c_0^n|^2|0_{\rm S}\rangle\langle 0_{\rm S}| + |c_1^n|^2|1_{\rm S}\rangle\langle 1_{\rm S}|\right]$.  Each value of $r$ defines a different quantum operation, $\rho^{n\prime}_{\rm S}(r)=\mathcal{Q}_r(\rho^n_{\rm S})$, with $Q_1$ the identity, and $Q_0$ completely dephasing the qubit.

The operation $Q_r$ can be implemented straightforwardly by a quantum CNOT gate and an auxiliary, making our Theorem open to experimental investigation \cite{VYH+2015}. With the auxiliary initially prepared in the state $\alpha|0_{\rm A}\rangle + \beta|1_{\rm A}\rangle$ acting as the target, we have $r=\left(\alpha^* \beta+ \alpha \beta^* \right)$. 

\textit{Non-diagonalisable outputs.}
When $0<r<1$, the outputs, $\rho^{n\prime}_{\rm S}(r)$, will not typically be simultaneously diagonalisable.  Consider $r=\frac{1}{\sqrt{2}}$ with two possible inputs: $|0_{\rm S}\rangle$ occurring with probability $0.3$, and $|+_{\rm S}\rangle = \frac{1}{\sqrt{2}}|0_{\rm S}\rangle + \frac{1}{\sqrt{2}}|1_{\rm S}\rangle$ occurring with probability $0.7$. After tracing away the auxiliary we find that the individual outputs are not simultaneously diagonalisable.  There is enough information to calculate the lower bound on $\epsilon_Q$ from Equation~(\ref{eqoffdiagM}) and this is found to be $0.0007\times kT$.  This can be compared with the bound $kT\ln 2\left[S(\rho_{{\rm S}}) - S(\rho'_{{\rm S}})\right] = -0.15 \times kT$. An overall negative value of $\Delta E_Q$ indicates that this operation may still be used to extract energy from the heat bath.

\textit{Diagonalisable outputs.}
For $Q_0$, the outputs will be simultaneously diagonalised in the basis $\{\ket{0_{\rm S}},\ket{1_{\rm S}}\}$.  We choose two pure state inputs: $\rho^1_{\rm S}$ with probability $p$, and $\rho^2_{\rm S}$ with probability $1-p$.  The average input density matrix is therefore $\rho_{\rm S} = p\rho^1_{\rm S} + (1-p)\rho^2_{\rm S}$. Figure \ref{F1} shows the minimum energy change in the heat bath with $p$ for four different pairs of inputs. We see that for some combinations of inputs there are no values of $p$ for which the operation can be thermodynamically reversible, and it necessarily requires an excess energy cost. For other inputs there are regions of $p$ where there is no excess cost and regions where there is a minimum non zero excess cost. Figure \ref{F2} shows, for a given $\rho^1_{\rm S}$, the regions of $\rho^2_{\rm S}$ on the Bloch sphere where a zero excess cost is possible for various values of $p$.  Further details of how to determine the minimum energy change in the heat bath can be found in Appendix \ref{ADP}.

The qubit dephasing example shows, in particular, that $\epsilon_{Q}=0$ is possible both with $\Delta E_Q<0$ and with non-trivial non-orthogonal input states, and that $\epsilon_{Q}>0$ is possible even when the output states share a common diagonalisation.

%%%%%%%%%%%%%%%%%%%%%%%%%%%%%%%%%%%%%%%%%%%%
\section{Discussion}
%%%%%%%%%%%%%%%%%%%%%%%%%%%%%%%%%%%%%%%%%%%%

The excess thermodynamic costs of quantum operations stem from the requirement that the operation should get the computation right for every individual input, and not just on average. For well known cases, including classical information processing, we have shown how thermodynamic reversibility can be reached as a special case.  However, for general quantum operations, with non-orthogonal signal states, an excess thermodynamic cost must be paid.

This might seem counter intuitive, since quantum computations can always be represented as unitary operations, which can always be run in reverse.  How can thermodynamic irreversibility arise with a reversible unitary operation?  Simply running the overall unitary operation in reverse does not just tidy up the auxiliary, returning it to its initial state, but it also undoes the computation, converting the output back to the input.  If we wish to retain the result of the computation, we cannot simply reverse the unitary operation. In the classical reversible computing model of Bennett \cite{BEN2}, we save a copy of the output before reversing the computation.  For quantum operations this is forbidden by the no-cloning theorem, and no quantum generalisation of Bennett's procedure is possible \cite{OM2}.

As long as we keep the results of the computation, we are left with the changes in the auxiliaries, representing spent resources and a cost to reset them to their initial states.  This is the stage at which thermodynamic costs are incurred.  We have shown that in general, for quantum operations this cost is necessarily in excess of the Landauer Limit given in Equation~(\ref{bound}).  Quantum computing requires a thermodynamically irreversible generation of heat.

%%%%%%%%%%%%%%%%%%%%%%%%%%%%%%%%%%%%%%%%%%%%
\section*{Acknowledgements}
%%%%%%%%%%%%%%%%%%%%%%%%%%%%%%%%%%%%%%%%%%%%
We would like to thank Chris Timpson and Oscar Dahlsten for discussions. This work was funded by the Templeton World Charity Foundation.

\appendix

%%%%%%%%%%%%%%%%%%%%%%%%%%%%%%%%%%%%%%%%%%%%%%
\section{Notation and formal set up}
%%%%%%%%%%%%%%%%%%%%%%%%%%%%%%%%%%%%%%%%%%%%%%
\label{NAFS}

Given an input $\rho^n_{\rm S}$ and an output ${\cal Q}(\rho^n_{\rm S}) = \rho^{n\prime}_{\rm S}$ of a quantum operation $Q$ on a quantum system ${\rm S}$ we can describe the operation as
\begin{align}
\rho^{n\prime}_{\rm S} = \tr_{\rm AB} [U \rho_{\rm S}^n\otimes \rho_{\rm A}\otimes \rho_{\rm B} U^{\dagger}],
\label{qop}
\end{align}
where the subscript on the trace indicates which components have been traced over. The unitary $U$ acts on an initial state $\rho_{\rm S}^n \otimes \rho_{\rm A}\otimes \rho_{\rm B}$ which includes, besides the system, a thermal environment in the form of a canonical heat bath, $\rho_{\rm B}$, and an auxiliary to be used as a catalyst $\rho_{\rm A}$.
Given that inputs $\rho^n_{\rm S}$ appear with probability $p_n$ we can define an average system input state and an average system output state by
\begin{align}
\rho_{\rm S} = \sum_n p_n \rho^n_{\rm S}  \quad \text{and} \quad
\rho'_{\rm S} = \sum_n p_n \rho^{n\prime}_{\rm S},
\nonumber
\end{align}
respectively.

The auxiliary is a resource which is spent by the operation. In order to account for this cost we demand that it must be returned to the same standard state $\rho_{\rm A}$ at the end of the computation
\begin{align}
\tr_{\rm SB} [U \rho_{\rm S}\otimes \rho_{\rm A}\otimes\rho_{\rm B} U^{\dagger}] = \rho_{\rm A}.
\nonumber
\end{align}
This ensures that the auxiliary can be reused in subsequent computations. We need only demand that the auxiliary is reset on average.

For convenience we define the following notation
\begin{align}
\rho_{\rm AB} &= \rho_{\rm A}\otimes\rho_{\rm B}, \nonumber \\
\rho' &= U (\rho_{\rm S}\otimes\rho_{\rm A}\otimes\rho_{\rm B}) U^{\dagger}, \nonumber \\
\rho_{\rm B}' &=\tr_{\rm SA} [U (\rho_{\rm S}\otimes\rho_{\rm A}\otimes\rho_{\rm B}) U^{\dagger}].
\nonumber
\end{align}

We can choose to work in a basis in which $\rho_{\rm S}$ is diagonal. In general $\rho'_{\rm S}$ will not be diagonal in this basis but we can perform a final rotation to make it so:
\begin{align}
R\rho'_{\rm S} R^{\dagger} =  \tr_{\rm AB} [(RU)\rho_{\rm S} \otimes \rho_{\rm AB} (RU)^{\dagger}] ,
\nonumber
\end{align}
or
\begin{align}
\bar{\rho}'_A = \tr_{\rm AB} [\bar{U} \rho_{\rm S}\otimes \rho_{\rm AB} \bar{U}^{\dagger}].
\nonumber
\end{align}
Now drop the bars from the notation and assume without loss of generality that the average system inputs and outputs are diagonalised in the same basis. Let us denote these diagonal basis states by $\{|\phi_i\rangle\}$. We can write
\begin{align}
\rho_{\rm S}  = \sum_i \lambda_i |\phi_i\rangle\langle \phi_i | \quad \text{and} \quad
\rho'_{\rm S} = \sum_i \lambda'_i |\phi_i\rangle\langle \phi_i |.
\nonumber 
\end{align}
In this basis an individual input and output state can in general be expressed as
\begin{align}
\rho^n_{\rm S} =  \sum_{ij} \mu^n_{ij} |\phi_i\rangle\langle \phi_j | \quad \text{and} \quad
\rho^{n\prime}_{\rm S} =  \sum_{ij} \mu^{n\prime}_{ij }|\phi_i\rangle\langle \phi_j |.
\label{inout}
\end{align}

We define
\begin{align}
A_{ij} = \langle \phi_i |U|\phi_j\rangle.
\nonumber
\end{align}
This is a bounded operator acting on the total environment of auxiliary and heat bath. That it is bounded can be demonstrated as follows: $U$ is a bounded operator since $|U|\psi\rangle| = ||\psi\rangle|$ for a state of system, auxiliary, and heat bath $|\psi\rangle$. This means that the product of the bounded operators $|\phi_i\rangle\langle\phi_i|U|\phi_j\rangle\langle\phi_j|$ is also bounded. In fact $||\phi_i\rangle\langle\phi_i|U|\phi_j\rangle\langle\phi_j|\psi\rangle| \leq  ||\psi\rangle|$. If we choose $|\psi\rangle = |\chi\rangle|\phi_j\rangle$ where $|\chi\rangle$ is an arbitrary total environment state we then find that $|A_{ij}|\chi\rangle| \leq ||\chi\rangle|$ so that $A_{ij}$ is a bounded operator.

We can express the unitary operator $U$ as
\begin{align}
U =  \sum_{ij}  |\phi_i\rangle\langle\phi_j| \otimes A_{ij}.
\label{uphi}
\end{align}
These operators must satisfy
\begin{align}
\langle \phi_i|UU^{\dagger} |\phi_j\rangle &= \sum_k A_{ik}A_{jk}^{\dagger} = \delta_{ij} \mathbb{1}_{\rm AB} ,
\label{g4}\\
\langle \phi_i|U^{\dagger}U |\phi_j\rangle &= \sum_k A_{ki}^{\dagger}A_{kj} = \delta_{ij} \mathbb{1}_{\rm AB}.
\label{g5}
\end{align}
The unitary operation on the complete state results in
\begin{align}
\langle \phi_i|\rho' |\phi_j\rangle
=\sum_k \lambda_k A_{ik} \rho_{\rm AB} A_{jk}^{\dagger}.
\label{g1}
\end{align}
It will also be useful to define the state $\rho_\star = \rho'_{\rm S}\otimes\rho_{\rm AB}$ which we can express as
\begin{align}
\langle \phi_i|\rho_{\star} |\phi_j\rangle
=\lambda'_i \delta_{ij} \rho_{\rm AB},
\label{g2}
\end{align}
and we denote
\begin{align}
\langle \phi_i|(\rho' - \rho_{\star}) |\phi_j\rangle = \Delta_{ij}.
\label{g3}
\end{align}
The result of the quantum operation on individual inputs (\ref{qop}) can be written using (\ref{inout})  and (\ref{uphi}) as
\begin{align}
\mu^{n\prime}_{kl} = \sum_{ij} \mu^n_{ij}\tr_{} [A_{ki}\rho_{\rm AB} A_{lj}^{\dagger}].
\label{indiv}
\end{align}
We will also use the notation
\begin{align}
q(kl|ij) =  \tr_{} [A_{ki}\rho_{\rm AB} A_{lj}^{\dagger}].
\label{qcoef}
\end{align}

%%%%%%%%%%%%%%%%%%%%%%%%%%%%%%%%%%%%%%%%%%%%%%
\section{Quantum thermodynamics}
%%%%%%%%%%%%%%%%%%%%%%%%%%%%%%%%%%%%%%%%%%%%%%
\label{Sqtd}
\begin{lemma}
\label{l0}
The energy change in the heat bath as a result of a quantum operation $Q$ on a system ${\rm S}$ implemented by some unitary $U \rho_{\rm S}\otimes \rho_{\rm A}\otimes \rho_{\rm B} U^{\dagger} = \rho'$, involving a resetting of the auxiliary system ${\rm A}$, is given by
\begin{align}
\frac{\Delta E_Q}{kT\ln 2} =\left[ S(\rho_{\rm S}) - S(\rho'_{\rm S})\right] + S(\rho'||\rho_{\star}).
\label{Eresult0}
\end{align}
where $\rho_{\star} = \rho'_{\rm S}\otimes\rho_{\rm A}\otimes\rho_{\rm B}$.
\end{lemma}
The Schumacher information measure is given by $S(\rho) = -\tr \left[\rho\log_2\rho\right]$ and we use a definition of the relative entropy using log base 2
\begin{align}
S(\rho||\sigma) = \tr \left[\rho\log_2\rho\right]-\tr \left[\rho\log_2\sigma\right].
\nonumber
\end{align}
\begin{proof}
The energy change in the bath is given by $\Delta E_Q = \tr  [H_{\rm B}(\rho_{\rm B}'-\rho_{\rm B})]$ where $H_{\rm B}$ is the heat bath Hamiltonian. A standard result is that \cite{PAV}
\begin{align}
\frac{\Delta E_Q}{kT\ln 2}
=S(\rho'_{\rm B}) - S(\rho_{\rm B})  + S(\rho'_{\rm B}||\rho_{\rm B}).
\nonumber
\end{align}
It follows that
\begin{align}
\frac{\Delta E_Q}{kT\ln 2}
=&S(\rho'_{\rm B}) - \left[S(\rho) - S(\rho_{\rm S}) - S(\rho_{\rm A})\right]  + S(\rho'_{\rm B}||\rho_{\rm B}) \nonumber \\
=&S(\rho'_{\rm B}) - \left[S(\rho') - S(\rho_{\rm S}) - S(\rho_{\rm A})\right]  + S(\rho'_{\rm B}||\rho_{\rm B}) \nonumber \\
=& \left[ S(\rho_{{\rm S}}) - S(\rho'_{{\rm S}})\right] + S(\rho'_{\rm B}||\rho_{\rm B})\nonumber\\
&\quad+\left[ S(\rho'_{{\rm S}}) + S(\rho_{\rm A}) + S(\rho'_{\rm B}) - S(\rho') \right] .
\label{DEcalc}
\end{align}
Since the last term in square brackets in the last line must be positive by subadditivity and the relative entropy term must be positive as a result of Klein's inequality, we must have
\begin{align}
\Delta E_Q \geq  kT\ln 2\left[ S(\rho_{\rm S}) - S(\rho'_{\rm S})\right] .
\label{LAND}
\end{align}
The operation satisfies thermodynamic reversibility if the inequality is saturated. This can be seen by forming
a second operation $\rho_{\rm S} = {\cal Q}'(\rho_{\rm S}')$ which restores the original average state. There must be a net heat gain in the heat bath if, for any part of the cycle, the equality in (\ref{LAND}) does not hold.

Write $\rho_\star = \rho_{{\rm S}}'\otimes \rho_{\rm A}\otimes\rho_{\rm B} $.
As a result of the fact that $\rho_\star$ factorises we find that
\begin{align}
S(\rho'||\rho_{\star}) = & \tr[\rho'\log_2\rho'] - \tr [\rho' \left( \log_2\rho'_{\rm S}\otimes {\mathbb 1}_{\rm A}\otimes{\mathbb 1}_{\rm B}\right)]
\nonumber\\
&- \tr [\rho' \left( {\mathbb 1}_{\rm S}\otimes \log_2\rho_{\rm A} \otimes{\mathbb 1}_{\rm B}\right)] \nonumber\\
&-  \tr[\rho'\left({\mathbb 1}_{\rm S}\otimes{\mathbb 1}_{\rm A}\otimes\log_2\rho_{\rm B}\right)]\nonumber\\
=& \tr[\rho'\log_2\rho'] - \tr[\rho'_{\rm S}\log_2\rho'_{\rm S}]
\nonumber\\ & - \tr[\rho_{\rm A}\log_2\rho_{\rm A}] - \tr[\rho'_{\rm B}\log_2\rho_{\rm B}]
\nonumber\\
=&S(\rho_{\rm B}'||\rho_{\rm B})\nonumber\\
&+\left[S(\rho_{\rm S}') + S(\rho_{\rm A}) + S(\rho_{\rm B}') - S(\rho') \right] ,
\label{REcalc}
\end{align}
so that from (\ref{DEcalc})
\begin{align}
\frac{\Delta E_Q}{kT\ln 2} =\left[ S(\rho_{{\rm S}}) - S(\rho'_{{\rm S}})\right] + S(\rho'||\rho_{\star}).
\nonumber
\end{align}
The quantity $S(\rho'||\rho_{\star}) \geq 0$ therefore encodes the extent to which the thermodynamic bound on the energy change (\ref{LAND}) is breached.
\end{proof}

The demand for $S(\rho'||\rho_{\star})=0$ requires from (\ref{REcalc}) that both $S(\rho') = S(\rho_{{\rm S}}') + S(\rho_{\rm A}) + S(\rho_{\rm B}')$ and $S(\rho_{\rm B}'||\rho_{\rm B}) = 0$. In particular the second of these conditions suggests an output state with $\rho_{\rm B}' = \rho_{\rm B}$ and therefore no energy change in the heat bath. In fact, in the limit where the dimension of the heat bath becomes large, it is possible to have $\rho_{\rm B}'$ sufficiently close to $\rho_{\rm B}$ such that $S(\rho'_{\rm B}||\rho_{\rm B})<s$ for some arbitrarily small $s$ whilst at the same time $\Delta E_Q > E$ for some fixed non zero $E$. To see this we write $\rho_{\rm B}' = \rho_{\rm B} + \varepsilon\Delta$ where $\Delta$ is a fixed traceless matrix and $\varepsilon$ is a small parameter. The change in energy of the heat bath is given by
\begin{align}
\Delta E_Q = \tr  [H_{\rm B}(\rho_{\rm B}'-\rho_{\rm B})] = \varepsilon \tr [ H_{\rm B} \Delta ].
\nonumber
\end{align}
Now suppose that we have $N$ identical independent copies of the same heat bath (this will also be a canonical state). If we further suppose that each copy undergoes the same uncorrelated change $\rho_{\rm B}' = \rho_{\rm B} + \varepsilon\Delta$, then to lowest order, the energy change of the $N$-copy heat bath is
\begin{align}
\Delta E_Q^{(N)} = N\Delta E_Q.
\nonumber
\end{align}
This means that we can take $\varepsilon \propto 1/N$ and as $N$ becomes large, $\Delta E_Q^{(N)}$ remains fixed.

Now consider the relative entropy between $\rho_{\rm B}$ and $\rho_{\rm B}'$
\begin{align}
S(\rho_{\rm B}'||\rho_{\rm B})  &=  \varepsilon\frac{d}{d\varepsilon}S(\rho_{\rm B} + \varepsilon\Delta||\rho_{\rm B})|_{\varepsilon = 0} + {\cal O}(\varepsilon^2) \nonumber\\
&= \varepsilon\tr\left[\frac{\Delta}{ \ln 2} + \Delta\log_2\rho_{\rm B} - \Delta\log_2\rho_{\rm B}\right]+ {\cal O}(\varepsilon^2) \nonumber\\
&=  {\cal O}(\varepsilon^2).
\nonumber
\end{align}
The correction to the relative entropy is at least quadratic in $\varepsilon$. The relative entropy for the $N$-copy heat bath is
\begin{align}
S(\rho_{\rm B}^{\prime\otimes N}||\rho_{\rm B}^{\otimes N})  = N S(\rho_{\rm B}'||\rho_{\rm B}).
\nonumber
\end{align}
Therefore, in the limit that $N$ becomes large and $\varepsilon \propto 1/N$ tends to zero, the relative entropy for the $N$-copy heat bath tends to zero. This simple example demonstrates that it is legitimate for the relative entropy to tend to zero whilst the energy change in the heat bath remains non zero.

\begin{corollary}
\label{C0}
The energy change in the heat bath is bounded by
\begin{align}
\Delta E_Q \geq {kT\ln2} \left[ S(\rho_{{\rm S}}) - S(\rho'_{{\rm S}})\right] + \frac{1}{2}kT|| \rho'-\rho_{\star}||_1^2.
\label{Eresult}
\end{align}
\end{corollary}
\begin{proof}
We use \cite{REbound1,REbound2}
\begin{align}
S(\rho||\sigma) \geq \frac{1}{2\ln 2}|| \rho-\sigma||_1^2.
\nonumber
\end{align}
The result follows from Equation~(\ref{Eresult0}).
\end{proof}

%%%%%%%%%%%%%%%%%%%%%%%%%%%%%%%%%%%%%%%%%%%%%%
\section{Standard thermodynamically reversible protocol}
%%%%%%%%%%%%%%%%%%%%%%%%%%%%%%%%%%%%%%%%%%%%%%
\label{ARev}
There is a standard protocol for reaching the bound given in Equation~(\ref{LAND}), appearing in various forms in the literature (\cite{Maroney2007a,VV} for examples), when the quantum operation is defined only to act upon an individual density matrix $\rho_{\rm S}'=\mathcal{Q}(\rho_{\rm S})$ of a system ${\rm S}$ (see also \cite{SH,AG2013} for alternative protocols to reach the bound).  Note that these protocols do not require that either the initial or final states are in thermal equilibrium.

The density matrix is initially in the state
\begin{align}
\rho_{\rm S}=\sum_i \lambda_i \proj{\phi_i},
\nonumber
\end{align}
where $\{\ket{\phi_i}\}$ are orthonormal eigenstates. The result of the operation $Q$ is a state
\begin{align}
\rho_{\rm S}'=\sum_i \lambda_i' \proj{\phi'_i},
\nonumber
\end{align}
for a possibly different set of orthonormal eigenstates $\{\ket{\phi'_i}\}$. We assume that both the initial and final sets of eigenstates are fully degenerate in energy with energy level zero. This avoids the complication of the system being used as a source or sink of energy.

{\it Step 1}.  Starting with $\rho_{\rm S}$, manipulate the energy levels $E_i$ of each eigenstate $\ket{\phi_i}$ until
\begin{align}
\lambda_i = \frac{e^{-\beta E_i}}{\sum_j e^{-\beta E_j}}.
\nonumber
\end{align}
where $\beta = 1/kT$. This has a mean cost $\Delta E_1=\sum_i \lambda_i E_i$.  The new density matrix will now be canonically distributed at temperature $T$.

{\it Step 2}. Bring the system into contact with a heat bath at temperature $T$.

{\it Step 3}. Slowly, isothermally, change the energy eigenstates until they satisfy
\begin{align}
\lambda'_i = \frac{e^{-\beta E'_i}}{\sum_j e^{-\beta E'_j}}.
\nonumber
\end{align}
The energy requirement for this is $\Delta E_2 = kT\ln\sum_j e^{-\beta E_j} - kT\ln\sum_j e^{-\beta E'_j}$.  The density matrix is now canonically distributed at temperature $T$, with eigenvalues $\{\lambda'_i\}$.

{\it Step 4}. Remove the system from contact with the heat bath.

{\it Step 5}. Change the energy levels of the eigenstates back to zero, with a mean cost $\Delta E_3=-\sum_i \lambda_i' E'_i$.

{\it Step 6}.  Perform a unitary rotation to the final eigenstates $\ket{\phi_i} \rightarrow \ket{\phi_i'}$.

The operation is complete with a cost of
\begin{align}
\Delta E_1 + \Delta E_2+\Delta E_3 = kT\ln2 \left[S(\rho_{\rm S}) - S(\rho_{\rm S}')\right].
\nonumber
\end{align}
It is worth noting that this protocol works even when the system is a joint system with non-trivial correlations or entanglement between subsystems \cite{VV}.

During the stages in which we manipulate the energy levels of the system we are assuming a time-dependent Hamiltonian $H(t)$ for the system and a standard definition of the mean rate of work given by $\tr [\rho \partial H/ \partial t]$. This presupposes some idealised work reservoir, capable of doing work on (or taking work from) the system by manipulation of the Hamiltonian.

We will denote this operation $T_{\rho_{\rm S}\rightarrow \rho_{\rm S}'}$. As shown in \cite{RW2014}, the theoretical energy cost can only be acheived in an aymptotic limit. Nevertheless, there is no physical principle which prevents us from getting arbitrarily close this value.

%%%%%%%%%%%%%%%%%%%%%%%%%%%%%%%%%%%%%%%%%%%%%%
\section{Some useful results}
%%%%%%%%%%%%%%%%%%%%%%%%%%%%%%%%%%%%%%%%%%%%%%
\label{SUR}

\begin{lemma}
The following relation holds
\begin{align}
 &(\lambda_j \lambda'_k  - \lambda_i\lambda'_l)q(kl|ij)  =
\nonumber\\
& \sum_m\left(\lambda'_k \tr [A_{ki}A_{mj}^{\dagger}\Delta_{ml}] -
\lambda'_l \tr [A_{mi}A_{lj}^{\dagger}\Delta_{km}]
\right).
\label{l1}
\end{align}
\end{lemma}

\begin{proof}
In order to prove this relation use Equations~(\ref{g1}), (\ref{g2}), and (\ref{g3}) to substitute for $\Delta_{ij}$. Then use the identity (\ref{g5}) along with the cyclic property of the trace. The coefficients $q(kl|ij)$ are given in (\ref{qcoef}).
\end{proof}

\begin{lemma}
\label{Lem3}
The following bound can be placed on the output state $\rho'$:
\begin{align}
||\rho'-\rho_{\star}||_1 \geq \frac{|\lambda_j\lambda'_k - \lambda_i\lambda'_l|}{\lambda'_k+\lambda'_l}
  \left| q(kl|ij) \right|  .
\label{lem2EQ}
\end{align}
\end{lemma}

\begin{proof}
We start by defining operators acting on the complete state of system, auxiliary, and heat bath
\begin{align}\nonumber
F_{ijklm} = |\phi_l\rangle\langle\phi_m |\otimes A_{ki}A_{mj}^{\dagger} ,
\end{align}
in terms of which we define
\begin{align}\nonumber
G^L_{ijkl} = \sum_m  F_{ijklm}; \quad
G^R_{ijkl} = \sum_m F_{jilkm}^{\dagger}.
\end{align}
We can then express (\ref{l1}) as
\begin{align}
(\lambda_j &\lambda'_k - \lambda_i\lambda'_l)q(kl|ij) = \nonumber\\
&\left(\lambda'_k \tr [G^L_{ijkl}(\rho'-\rho_{\star})]- \lambda'_l \tr [G^R_{ijkl}(\rho'-\rho_{\star})]
\right).
\label{el3}
\end{align}
Let us define some arbitrary basis states $\{|\psi_x\rangle\}$ for the complete system, auxiliary, and heat bath. We can write
\begin{align}
 G^{L/R}_{ijkl} &= \sum_{xy} g^{L/R}_{xy}|\psi_x\rangle\langle \psi_y | ; \nonumber\\
(\rho' - \rho_{\star}) &= \sum_{xy} r_{xy}|\psi_x\rangle\langle \psi_y |. \nonumber
\end{align}
Then
\begin{align}
\left|\tr [G^L_{ijkl} (\rho' - \rho_{\star})]\right| &= \left|\sum_{xy} g^L_{xy}r_{yx}\right|  \nonumber\\
&\leq \sum_x \left|\sum_y g^L_{xy}r_{yx}\right| \nonumber\\
&\leq \sum_x \left(\sum_y |g^L_{xy}|^2\right)^{\frac{1}{2}}\left(\sum_z |r_{zx}|^2\right)^{\frac{1}{2}} ,
\label{G1eq}
\end{align}
where the second line follows from the triangle inequality and the third line follows from the Cauchy-Schwartz inequality. Similarly
\begin{align}
\left|\tr [G^R_{ijkl} (\rho' - \rho_{\star})]\right| \leq \sum_x \left(\sum_y |g^R_{yx}|^2\right)^{\frac{1}{2}}\left(\sum_z |r_{xz}|^2\right)^{\frac{1}{2}}.
\label{G2eq}
\end{align}
Next we show that
\begin{align}
\langle \psi_x| G^L_{ijkl}G^{L\dagger}_{ijkl} |\psi_x\rangle  = \sum_y |g^L_{xy}|^2 \leq 1,
\label{rdiag}
\end{align}
and
\begin{align}
\langle \psi_x| G^{R\dagger}_{ijkl}G^{R}_{ijkl} |\psi_x\rangle  = \sum_y |g^R_{yx}|^2 \leq 1.
\label{ldiag}
\end{align}
The operator products can be written as
\begin{align}
G^L_{ijkl}G^{L\dagger}_{ijkl} &= \sum_m F_{ijklm}F^{\dagger}_{ijklm} ; \nonumber\\
G^{R\dagger}_{ijkl}G^{R}_{ijkl} &= \sum_m F_{jilkm}F^{\dagger}_{jilkm} . \nonumber
\end{align}
We also have
\begin{align}\nonumber
\delta_{jj'}\delta_{kk'} \mathbb{1} = \sum_{ilm}F_{ijklm} F_{ij'k'lm}^{\dagger},
\end{align}
which follows from (\ref{g4}) and (\ref{g5}). Since the identity is expressed as a sum of positive operators and since $G^L_{ijkl}G^{L\dagger}_{ijkl}$ is a sum of a subset of these positive operators, it must be the case that the diagonal elements of $G^L_{ijkl}G^{L\dagger}_{ijkl}$ in the arbitrary basis $\{|\psi_x\rangle\}$ must belong to the range $[0,1]$ from which (\ref{rdiag}) follows. Similarly for (\ref{ldiag}). We can therefore write, from (\ref{G1eq}) and (\ref{G2eq}),
\begin{align}
\left|\tr [G^{L}_{ijkl} (\rho' - \rho_{\star})]\right| &\leq \sum_x \left(\sum_z |r_{zx}|^2\right)^{\frac{1}{2}} ,
\nonumber \\
\left|\tr [G^{R}_{ijkl} (\rho' - \rho_{\star})]\right| &\leq \sum_x \left(\sum_z |r_{xz}|^2\right)^{\frac{1}{2}} .\nonumber
\end{align}
If we choose the basis states $\{|\psi_x\rangle\}$ such that $(\rho' - \rho_{\star})$ is diagonal we have
\begin{align}\nonumber
\left|\tr [G^{L/R}_{ijkl} (\rho' - \rho_{\star})]\right| \leq ||\rho' - \rho_{\star}||_1.
\end{align}
Finally from (\ref{el3}), using the triangle inequality
\begin{align}
\left|\lambda_j \lambda'_k - \lambda_i\lambda'_l\right| \left| q(kl|ij)\right| & \leq
\nonumber\\
\lambda'_k \left|\tr [G^{R}_{ijkl}(\rho'-\rho_{\star})]\right| &+
 \lambda'_l \left|\tr [G^{L}_{ijkl}(\rho'-\rho_{\star})]\right|
\nonumber\\
&\leq (\lambda'_k+\lambda'_l)||\rho'-\rho_{\star}||_1, \nonumber
\end{align}
which completes the proof.
\end{proof}

\begin{lemma}[Special case when dimension $d=2$]
\label{Lem2}
When the dimension of the system space $d=2$, the following bound can be placed on the output state $\rho'$ for $i\neq j$:
\begin{align}\nonumber
||\rho'-\rho_{\star}||_1 \geq |\lambda_i - \lambda_j|
\left| q(kk|ij) \right| .
\end{align}
\end{lemma}

\begin{proof}
When $k = l$ Equation~(\ref{l1}) reduces to
\begin{align}
(\lambda_i & - \lambda_j) q(kk|ij) =  \nonumber\\
&
\sum_{m\neq k}\left(\tr [A_{mi}A_{kj}^{\dagger}\Delta_{km}] -  \tr [A_{ki}A_{mj}^{\dagger}\Delta_{mk}]
\right).
\label{l2a}
\end{align}
If $i=j$ both sides of this equation are zero. If $i\neq j$ we can define operators acting on the complete system, auxiliary, and heat bath
\begin{align}\nonumber
H_{ijkm} =  |\phi_m\rangle\langle\phi_k |\otimes A_{mi}A_{kj}^{\dagger}  -  |\phi_k\rangle\langle\phi_m |\otimes A_{ki}A_{mj}^{\dagger},
\end{align}
and then express (\ref{l2a}) as
\begin{align}
(\lambda_i - \lambda_j)q(kk|ij) =
\sum_{m\neq k} \tr [H_{ijkm} (\rho' - \rho_{\star})].
\label{l2b}
\end{align}
Define arbitrary basis states $\{|\psi_x\rangle\}$ for the complete system, auxiliary, and heat bath. We can write
\begin{align}
H_{ijkm} &= \sum_{xy} h_{xy}|\psi_x\rangle\langle \psi_y |; \nonumber\\
(\rho' - \rho_{\star}) &= \sum_{xy} r_{xy}|\psi_x\rangle\langle \psi_y |. \nonumber
\end{align}
Then
\begin{align}
\left|\tr [H_{ijkm} (\rho' - \rho_{\star})]\right| &= \left|\sum_{xy} h_{xy}r_{yx}\right|  \nonumber\\
&\leq \sum_x \left|\sum_y h_{xy}r_{yx}\right| \nonumber\\
&\leq \sum_x \left(\sum_y |h_{xy}|^2\right)^{\frac{1}{2}}\left(\sum_z |r_{zx}|^2\right)^{\frac{1}{2}} .
\label{Heq}
\end{align}
Next we show that
\begin{align}
\langle \psi_x|  H_{ijkm}  H_{ijkm} ^{\dagger} |\psi_x\rangle  = \sum_y |h_{xy}|^2 \leq 1.
\label{hdiag}
\end{align}
The operator product can be written
\begin{align}
H_{ijkm}H^{\dagger}_{ijkm} =  F_{ijmmk}F^{\dagger}_{ijmmk} + F_{ijkkm}F^{\dagger}_{ijkkm}.
\label{opprod}
\end{align}
We also have
\begin{align}
\delta_{jj'} \mathbb{1} = \sum_{ikm}F_{ijkkm} F_{ij'kkm}^{\dagger}.
\label{pos2}
\end{align}
Since (\ref{opprod}) is the sum of two constituent positive terms of the sum of positive operators (\ref{pos2}) then the same argument as that used in Lemma \ref{Lem3} can be used to show that (\ref{hdiag}) holds. It follows from (\ref{Heq}) that
\begin{align}\nonumber
\left|\tr [H_{ijkm} (\rho' - \rho_{\star})]\right| \leq \sum_x \left(\sum_z |r_{zx}|^2\right)^{\frac{1}{2}},
\end{align}
and if we choose the basis states $\{|\psi_x\rangle\}$ such that $(\rho' - \rho_{\star})$ is diagonal we have
\begin{align}\nonumber
\left|\tr [ H_{ijkm} (\rho' - \rho_{\star})]\right| \leq ||\rho' - \rho_{\star}||_1.
\end{align}
Inserting into (\ref{l2b}) we have
\begin{align}
\left|\lambda_i - \lambda_j\right|\left| q(kk|ij)\right| &\leq \sum_{m\neq k}
\left|\tr [H_{ijkm}(\rho'-\rho_{\star})]\right|
\nonumber\\
&\leq (d-1) ||\rho'-\rho_{\star}||_1,\nonumber
\end{align}
and choosing $d=2$ completes the proof.
\end{proof}

Note that Lemma \ref{Lem3} can also be applied to the case of $k=l$ and $i\neq j$ for general $d$ with the result
\begin{align}\nonumber
||\rho'-\rho_{\star}||_1 \geq \frac{|\lambda_i - \lambda_j|}{2} |q(kk|ij)|,
\end{align}
which is half the size of the bound of Lemma \ref{Lem2} in the case of $d=2$.

%%%%%%%%%%%%%%%%%%%%%%%%%%%%%%%%%%%%%%%%%%%%%%
\section{Proof of Theorem \ref{th1}}
\label{secPROOF}
%%%%%%%%%%%%%%%%%%%%%%%%%%%%%%%%%%%%%%%%%%%%%%
Theorem \ref{th1} is stated in the main text. We repeat it here for convenience
\setcounter{theorem}{0}
\begin{theorem}
If there exists a stochastic map, $P_Q(k|i)$ with
\begin{align}\nonumber
\sum_k P_Q(k|i) = 1 \quad \text{and} \quad P_Q(k|i) \geq 0 \; \forall i,k,
\end{align}
such that
\begin{align}
\mu^{n\prime}_{kl} = \delta_{kl} \sum_i P_Q(k|i) \mu^n_{ii},
\label{revconSI}
\end{align}
for all $n$, then the minimum thermodynamic cost of the operation
\begin{align}
\Delta E_Q \geq kT \ln 2 \left[ S(\rho_{\rm S}) - S(\rho'_{\rm S})\right]+\epsilon_Q,
\nonumber
\end{align}
can approach $\epsilon_Q=0$. Otherwise, provided there are no symmetries of the form $\lambda_i/\lambda_j=\lambda'_k/\lambda'_l$ where $i\neq j$ or $k\neq l$, then necessarily $\epsilon_Q>0$.
\end{theorem}

\begin{proof}
In general the output state of the quantum operation is given by (\ref{indiv})
\begin{align}
\mu^{n\prime}_{kl} = \sum_{ij} \mu^n_{ij}q(kl|ij).
\label{EQrewrite1}
\end{align}
By Corollary \ref{C0} the thermodynamic bound requires that $||\rho'-\rho_{\star}||_1 = 0$.
By Lemma \ref{Lem3} this requires that if $|\lambda_j\lambda'_k - \lambda_i\lambda'_l|$ is non zero then
\begin{align}
q(kl|ij) = 0.
\nonumber
\end{align}
It is helpful to separate Equation~(\ref{EQrewrite1}) into three terms:
\begin{align}
\mu^{n\prime}_{kl}=&\delta_{kl} \sum_i q(kk|ii) \mu^n_{ii} + \delta_{kl} \sum_{i,j\neq i} q(kk|ij) \mu^n_{ij}
\nonumber\\
&+ (1-\delta_{kl})\sum_{ij} q(kl|ij) \mu^n_{ij}.
\label{3split}
\end{align}
The first and second terms only deal with the diagonal elements of the output signal states in the basis of the average output density matrix.  The third term deals with the off-diagonal elements.

By Lemma \ref{Lem3}, the first sum implies only $|| \rho'-\rho_\star||_1\geq 0$, as $\lambda_i\lambda'_k - \lambda_i\lambda'_k=0$ for $q(kk|ii)$. If Equation ~(\ref{revconSI}) holds then $P_Q(k|i)=q(kk|ii)$ gives an implementation in which all other terms are zero, and there is no excess cost. The numbers $q(kk|ii)$ are real, non negative and satisfy
\begin{align}
\sum_k q(kk|ii) = \sum_k \tr [A_{ki}^{\dagger} A_{ki}\rho_{\rm AB}] = 1,
\nonumber
\end{align}
by Equation~(\ref{g5}).

If Equation ~(\ref{revconSI}) does not hold, then some values of $q(kl|ij)$ with $i\neq j$ or $k\neq l$ must be non-zero, in any implementation. For a given implementation, the largest value of the lower bound of Equation~(\ref{lem2EQ}) across all $ijkl$ determines an excess thermodynamic cost. If there are no symmetries of the form $\lambda_i/\lambda_j=\lambda'_k/\lambda'_l$ with $i\neq j$ or $k\neq l$, then this lower bound must be greater than zero. The minimum such value across all possible implementations gives a lower bound for $\epsilon_Q>0$.
\end{proof}

It should be noted that the existence of some values of $ijkl$ for which $\lambda_i/\lambda_j=\lambda'_k/\lambda'_l$ holds is not sufficient for $\epsilon_Q=0$.  To avoid an excess cost from our Theorem through the exceptional symmetry condition, it must be possible to implement the operation in such a way that $q(kl|ij)=0$ for every $ijkl$ value for which $\lambda_i/\lambda_j \neq \lambda'_k/\lambda'_l$.

When the output signal states do not share a common diagonalisation, we can give a general lower bound for $\epsilon_Q$.  The third sum in Equation~(\ref{3split}) is non-zero whenever there exists $k \neq l$ for which $\mu^{n\prime}_{kl}>0$. For these cases, from Equation~(\ref{EQrewrite1}), using the triangle inequality
\begin{align}
|\mu^{n\prime}_{kl}| \leq \sum_{ij} |q(kl|ij)|| \mu^n_{ij}|.
\nonumber
\end{align}
From Equation~(\ref{lem2EQ}) this implies
\begin{align}
|| \rho'-\rho_{\star}||_1  \geq
\max_{n,k,l\neq k}\left\{
\frac
{|\mu^{n\prime}_{kl}|/(\lambda'_k+\lambda'_l)}
{\sum_{ij} |\mu^n_{ij}|/|\lambda_j\lambda'_k - \lambda_i\lambda'_l|}
\right\}>0,
\nonumber
\end{align}
where the maximum is taken across all $k$ and $l$ such that $k\neq l$, and across all inputs $n$.
As the right hand side only involves terms independent of the specific implementation, then this bound must hold for all implementations and so by Corollary \ref{C0},
\begin{align}
\epsilon_Q  \geq \frac{1}{2}kT
\max_{n,k,l\neq k}\left\{
\frac
{|\mu^{n\prime}_{kl}|/(\lambda'_k+\lambda'_l)}
{\sum_{ij} |\mu^n_{ij}|/|\lambda_j\lambda'_k - \lambda_i\lambda'_l|}
\right\}^2. \nonumber
\end{align}
If all $\mu^{n\prime}_{kl}=0$, with $k \neq l$, but Equation~(\ref{revconSI}) does not hold, then
\begin{align}
\mu^{n\prime}_{kl}=\delta_{kl} \sum_i q(kk|ii) \mu^n_{ii} + \delta_{kl} \sum_{i,j\neq i} q(kk|ij) \mu^n_{ij},
\label{EQdiag}
\end{align}
and any implementation must have values of $|q(kk|ij)| > 0 $ with $i \neq j$. We can find complex coefficients $q(kk|ij)$, which satisfy Equation~(\ref{EQdiag}), while minimising the largest value of the lower bound of (\ref{lem2EQ}), $\frac{1}{2}|\lambda_i - \lambda_j| \left| q(kk|ij) \right|$. This must be done using numerical optimisation techniques on a case-by-case basis. Given the set $\{\{\lambda_i\},\{\lambda'_k\},\{\mu_{ij}^n\},\{\mu_{kl}^{n\prime}\}\}$, this establishes a minimum value for $|| \rho'-\rho_{\star}||_1$ and therefore, by Corollary \ref{C0}, a minimum excess cost across all possible implementations, so again $\epsilon_Q>0$.

%%%%%%%%%%%%%%%%%%%%
\section{Thermodynamically reversible protocol for special cases}
\label{proto}
%%%%%%%%%%%%%%%%%%%%

We now demonstrate a protocol which can achieve $\epsilon_Q=0$ for the special class of quantum operations which satisfy Equation~(\ref{revconSI}). For these operations there exists a stochastic map $P_Q(k|i)$ such that the input and output states of a quantum operation are related by $\mu^{n\prime}_{kl} = \delta_{kl} \sum_{i} P_Q(k|i) \mu^n_{ii}$. 

Let the initial state of the auxiliary be $\proj{0_{\rm A}}$, so the combined system and auxiliary is initially
\begin{align}
\rho_{\rm S}^n\otimes\rho_{\rm A}=\sum_{ij} \mu^n_{ij}\kb{\phi_i}{\phi_j}\otimes\proj{0_{\rm A}}.
\nonumber
\end{align}

{\it Step 1}. Use the auxiliary to measure the input system in the basis of the average input density matrix.  Correlate the auxiliary to the system using a unitary with $\ket{\phi_i}\ket{0_{\rm A}} \rightarrow \ket{\phi_i}\ket{i_{\rm A}}$.  The joint state is now
\begin{align}
\sum_{ij} \mu^n_{ij} \kb{\phi_i}{\phi_j}\otimes\kb{i_{\rm A}}{j_{\rm A}}.
\nonumber
\end{align} 

{\it Step 2}. The operation $T_{\proj{\phi_i}\rightarrow \rho^{(i)}_{\rm S}}$ (see Appendix \ref{proto}) performs a thermodynamically optimal conversion of $\proj{\phi_i}$ to $\rho^{(i)}_{\rm S}=\sum_k P_Q(k|i) \proj{\phi_k}$.  The conditional operation $\sum_i  T_{\proj{\phi_i} \rightarrow \rho^{(i)}_{\rm S}} \otimes \proj{i_{\rm A}} $ puts the joint system in the state
\begin{align}
\sum_{ki}P_Q(k|i)\mu^n_{ii} \proj{\phi_k} \otimes \proj{i_{\rm A}},
\nonumber
\end{align}
transferring mean heat to the heat bath $\Delta E_{1}=-k T \ln 2 \sum_{ni} p_n \mu^n_{ii} S(\rho^{(i)}_{\rm S})$. The average state takes the form
\begin{align}
\sum_{ki}P_Q(k|i) \lambda_i  \proj{\phi_k} \otimes \proj{i_{\rm A}}.
\nonumber
\end{align}

{\it Step 3}. We now exploit the correlations between system and auxiliary to optimally reset the auxiliary.  The operation $T_{\rho^{(k)}_{\rm A}\rightarrow \proj{0_{\rm A}}}$ performs an optimal conversion of $\rho^{(k)}_{\rm A}=\sum_{i} P_Q(k|i)\lambda_i/\left(\sum_j P_Q(k|j)\lambda_j\right) \proj{i_{\rm A}}$ to $\proj{0_{\rm A}}$, where $P_Q(k|i)\lambda_i/\left(\sum_j P_Q(k|j)\lambda_j\right)$ is the mean conditional probability of the auxiliary state being $|i_{\rm A}\rangle$ given that the system state is $|\phi_k\rangle$.  The conditional operation $\sum_k \proj{\phi_k} \otimes T_{\rho^{(k)}_{\rm A}\rightarrow \proj{0_{\rm A}}}$ resets the auxiliary to $\proj{0_{\rm A}}$ while transferring mean heat $\Delta E_{2}=k T \ln 2 \sum_{ki} P_Q(k|i)\lambda_i S(\rho^{(k)}_{\rm A})$.

The joint state is now
\begin{align}
\sum_{ki}P_Q(k|i)\mu^n_{ii} \proj{\phi_k} \otimes \proj{0_{\rm A}},
\nonumber
\end{align}
and provided the conditions of the theorem hold, the system is in the correct output state, up to a unitary.  The average heat production is readily confirmed to be $\Delta E_Q  =\Delta E_1 +\Delta E_2 = k T \ln 2 [S(\rho_{\rm S})-S(\rho_{\rm S}')]$ thus proving that $\epsilon_Q = 0$.  It is important to note that, if the conditions of the theorem do not apply, this protocol not only fails to be optimal, but fails to correctly implement the operation at all, i.e. $\rho_{\rm S}^n$ is not mapped to the correct $\rho^{n\prime}_{\rm S}$.

%%%%%%%%%%%%%%%%%%%%%%%%%%%%%%%%%%%%%%%%%%%%%%
\section{Qubit dephasing example}
%%%%%%%%%%%%%%%%%%%%%%%%%%%%%%%%%%%%%%%%%%%%%%
\label{ADP}

We analyse the specific example of a qubit dephasing operation. This can be implemented using a CNOT. For this operation an input of the form $\rho^n_{\rm S} = |v^n_{\rm S}\rangle\langle v^n_{\rm S}|$ with $|v^n_{\rm S}\rangle = c^n_0|0\rangle + c^n_1|1\rangle$ gives the output $\rho^{n\prime}_{\rm S} = |c^n_0|^2|0\rangle\langle 0 | +|c^n_1|^2|1\rangle\langle 1|$. All outputs are simultaneously diagonalisable in the $\{|0\rangle,|1\rangle\}$ basis. We choose two pure state inputs $\rho^1_{\rm S}$ and $\rho^2_{\rm S}$ such that the average input state is $\rho_{\rm S} = p\rho^1_{\rm S} +(1-p) \rho^2_{\rm S}$ for some probability $p$.

All of the elements $\mu^1_{ij},\mu^2_{ij},\mu^{1\prime}_{ij},\mu^{2\prime}_{ij}$ are specified by the operation together with the form and relative frequency of the inputs. To find $\mu^1_{ij},\mu^2_{ij}$ we determine the diagonal basis of $\rho_{\rm S}$ and express the individual inputs in this basis
\begin{align}
\rho^1_{\rm S} =& \mu^1_{11}|\phi_1\rangle\langle \phi_1|+\mu^1_{12}|\phi_1\rangle\langle \phi_2| \nonumber\\
&\quad+\mu^1_{21}|\phi_2\rangle\langle \phi_1|+\mu^1_{22}|\phi_2\rangle\langle \phi_2|, \nonumber \\
\rho^2_{\rm S} =& \mu^2_{11}|\phi_1\rangle\langle \phi_1|+\mu^2_{12}|\phi_1\rangle\langle \phi_2| \nonumber \\
&\quad+\mu^2_{21}|\phi_2\rangle\langle \phi_1|+\mu^2_{22}|\phi_2\rangle\langle \phi_2|.
\nonumber
\end{align}
The output elements are
\begin{align}
\mu^{n\prime}_{11} &= |c^n_0|^2, \nonumber\\
\mu^{n\prime}_{22} &= |c^n_1|^2, \nonumber\\
\mu^{n\prime}_{kl} & = 0 \text{ for } k\neq l.\nonumber
\end{align}
The elements must satisfy Equation~(\ref{EQrewrite1}) for some set of coefficients $q(kl|ij)$. The off-diagonal elements of the outputs can be solved by choosing $q(kl|ij) = 0$ for $k\neq l$, implying no constraint on excess heat cost by Lemma \ref{Lem3}. For the diagonal elements of the outputs we denote $q(kk|ii) = q_{ki}$, and $q(11|12) =w = -q(22|12)$, $q(11|21) =  w^* = -q(22|21)$, so that from Equation~(\ref{EQrewrite1}) we can write
\begin{align}
\mu^{1\prime}_{11} &= |c^1_0|^2 = \mu^1_{11} q_{11} + \mu^1_{12}w + \mu^1_{21}w^* + \mu^1_{22}q_{12}, \nonumber \\
\mu^{2\prime}_{11} &= |c^2_0|^2 = \mu^2_{11} q_{11} + \mu^2_{12}w + \mu^2_{21}w^* + \mu^2_{22}q_{12} , \nonumber
\end{align}
and
\begin{align}
\mu^{1\prime}_{22} &= |c^1_1|^2 = \mu^1_{11} q_{21} - \mu^1_{12}w - \mu^1_{21}w^* + \mu^1_{22}q_{22}, \nonumber\\
\mu^{2\prime}_{22} &= |c^2_1|^2 = \mu^2_{11} q_{21} - \mu^2_{12}w - \mu^2_{21}w^* + \mu^2_{22}q_{22} .
\nonumber
\end{align}
Since $|c^n_0|^2+|c^n_1|^2 = 1$, $\mu^n_{11}+\mu^n_{22} = 1$, and $q_{1k} +q_{2k}=1$, these pairs of equations are equivalent, so we focus on the first pair. The remaining numerical problem is to find the $w$ with the smallest magnitude which can solve this pair of equations with $0\leq q_{11}, q_{12} \leq 1$. This is done by scanning complex values of $w$ of increasing magnitude and solving for $q_{11}$ and $q_{12}$ until they both fit within the required range. Once we have found $|w|_{min}$ we can use Lemma \ref{Lem2} to bound  $||\rho'-\rho_{\star}||_1$:
\begin{align}
||\rho'-\rho_{\star}||_1\geq |\lambda_1-\lambda_2| |w|_{min},
\nonumber
\end{align}
where $\lambda_i$ are the eigenvalues of $\rho_{\rm S}$. Finally we can use Equation~(\ref{Eresult}) to bound the energy cost:
\begin{align}
\Delta E_{Q_0} \geq kT\ln 2\left[ S(\rho_{\rm S}) - S(\rho'_{\rm S})\right] + \frac{1}{2} kT (\lambda_1-\lambda_2)^2 |w|^2_{min}.
\nonumber
\end{align}
Examples are given in Figure \ref{F1}.

%%%%%%%%%%%%%%%%%%%%%%%%%%%%%%%%%%%%%%%%%%%%%%%%%%%%%%%%%%%%%%%%%%%%
%%%%%%%%%%%%%%%%%%%%%%%%%%%%%%%%%%%%%%%%%%%%%%%%%%%%%%%%%%%%%%%%%%%%

\end{document}